\documentclass[12pt]{article}
\usepackage[small]{titlesec}
\usepackage{amsmath,amssymb,amsthm,graphicx,setspace,paralist,booktabs,rotating,subcaption,float,color}
\usepackage{graphicx}
\usepackage{enumerate}
\usepackage{url} % not crucial - just used below for the URL 
\usepackage{amssymb, amsmath, amsthm}
\usepackage{breakcites}

\usepackage{bigstrut} %needed to keep blockarrays from being crunched
\usepackage{blkarray}
\usepackage{booktabs}
\usepackage{enumitem} % detailed list control
\usepackage{float} % for 'H' forced placement of floats
\usepackage{graphicx}
\usepackage{subcaption}
\usepackage{mwe}
\usepackage{comment}
\usepackage[colorlinks=true, allcolors=LinkBlue, backref=page]{hyperref}
\usepackage{lmodern}
\usepackage{mathrsfs} % for mathscr
\usepackage{multirow}
\usepackage{scalerel}
\usepackage{setspace}
\usepackage{soul} % for \st = strikethrough
\usepackage{subcaption}
\usepackage{titling}
\usepackage{verbatim}
\usepackage{xcolor}
\usepackage{xfrac}
\usepackage{algpseudocode}
\usepackage{bbm}
\usepackage{setspace}
\usepackage{amsfonts}
\usepackage{outlines}

\AtBeginDocument{
	\abovedisplayskip=5pt plus 2pt minus 2pt
	\belowdisplayskip=\abovedisplayskip
	\abovedisplayshortskip=2pt plus 2pt minus 2pt
	\belowdisplayshortskip=\belowdisplayskip}

\titlespacing*{\section}{0pt}{*2}{*1}
\titlespacing*{\subsection}{0pt}{*2}{*1} 
\setlist{noitemsep, topsep=0pt} % for enumitem
\definecolor{LinkBlue}{rgb}{.15, .25, .85} %for hyperref

\binoppenalty=\maxdimen % prevent linebreaks after binary operators, usually 700
\relpenalty=\maxdimen % prevent linebreaks after relations, usually 500

\usepackage{algorithm}
\makeatletter

\newfloat{algorithm}{tbp}{loa}
\providecommand{\algorithmname}{Algorithm}
\floatname{algorithm}{\protect\algorithmname}

\newtheorem{theorem}{Theorem}

\newtheorem{lemma}{Lemma}
\newtheorem{remark}{Remark}

\newtheorem{prop}{Proposition}

\def \bP {\mathbb{P}}
\def \bE {\mathbb{E}}

\usepackage{xspace}

\newcommand{\prob}[1]{\mathbb{P}\left[#1\right]}

\newcommand{\Ccen}{C_{\mathsf{Cen},1+\delta}}

\newcommand{\Clip}{\mathsf{Clip}}
\newcommand{\Qmc}{\mathcal{A}_{\mathsf{QMC}}}
\newcommand{\Cmc}{\mathcal{A}_{\mathsf{CMC}}}

\newcommand{\calA}{{\mathcal{A}}}

\newcommand{\calN}{{\mathcal{N}}}
\newcommand{\calO}{{\mathcal{O}}}
\newcommand{\calP}{{\mathcal{P}}}

\begin{document}

	\title{Quadratic Speed-up in Infinite Variance Quantum Monte Carlo}
 \author{Jose Blanchet,  Mario Szegedy, Guanyang Wang}
 \maketitle
\def\spacingset#1{\renewcommand{\baselinestretch}%
{#1}\small\normalsize} \spacingset{1}
\spacingset{1.2} 

\begin{abstract}
    In this study, we give an extension of Montanaro's  \cite{montanaro2015quantum}
   quantum Monte Carlo method, tailored for computing expected values of random variables that exhibit infinite variance. This addresses a challenge in analyzing heavy-tailed distributions, which are commonly encountered in various scientific and engineering fields. Our quantum algorithm efficiently estimates means for variables with a finite $(1+\delta)^{\text{th}}$ moment, where $\delta$ lies between 0 and 1. It provides a quadratic speedup over the classical Monte Carlo method in both the accuracy parameter $\epsilon$ and the specified moment of the distribution. We establish both classical and quantum lower bounds, showcasing the near-optimal efficiency of our algorithm among quantum methods. Our work focuses not on creating new algorithms, but on analyzing the execution of existing algorithms with available additional information about the random variable. Additionally, we categorize these scenarios and demonstrate a hierarchy in the types of supplementary information that can be provided.
\end{abstract}

\section{Introduction}
Computing expected values of random variables is important in many scientific and engineering fields. Classical Monte Carlo methods are widely used to approximate these values, assuming the random variable has a finite variance. However, in certain applied settings with heavy-tailed variables, infinite variance models may be more appropriate. These models are useful when the largest fluctuations in data grow faster than the square root of the number of observations collected. This can depend on the time scales of data collection and decision-making. For example, in finance, infinite variance models may be more suitable for short time scales (e.g., days), while finite variance models may be better for longer time scales (e.g., months); see \cite{gencay2001}. The work of \cite{GS2010} provides a multi-scale approach to build models that explain these empirical findings.

Heavy-tailed phenomena can explain empirical findings in internet traffic, earth sciences, and other areas, leading to long-range dependence and persistent correlations in time series data. In the context of internet traffic, a server experiences packet requests of users navigating the internet at a high frequency, probably in the order of seconds or even faster, whereas users consume information requested in the order of many minutes or perhaps even hours. This creates on-and-off interactions between the user and the server which relative to the server's time scale of operation are better modeled as very heavy-tailed random times \cite{taqqu1995estimators}. These models were pioneered by Mandelbrot \cite{mandelbrot1963} and have been studied in various fields; see \cite{pipiras_taqqu_2017} and \cite{Watkins_2019}.
Heavy-tailed distributions (potentially with infinite variance), however, can also pose challenges for statistical estimation, leading to inaccurate predictions and risk assessments. Consequently, accurately computing expected values for such distributions is a critical practical necessity.

Quantum computing offers potential speed-ups for various computational problems, including Monte Carlo methods. The work by \cite{montanaro2015quantum} introduced a quantum Monte Carlo algorithm with quadratic speed-up over classical methods for estimating the expected value of a random variable with finite variance, which generalizes the algorithm of \cite{heinrich2003monte} for mean estimation under the uniform distribution. Montanaro's algorithm relies on phase estimation 
\cite{brassard1998quantum, brassard2002quantum} which is in turn
grounded in Grover's algorithm \cite{grover1996fast}. It has also been applied in quantum Markov Chain Monte Carlo \cite{childs2022quantum} and computational finance \cite{bouland2020prospects}. 
Hamoudi and Magniez \cite{hamoudi2019quantum} improve Montanaro's algorithm in terms of relative error, assuming that an upper and lower bound for the expected value is known to the users. The thesis of Hamoudi \cite{Hamoudi2021Quantum} also improves Montanaro's algorithm regarding additive errors, eliminating the need for a pre-determined upper bound on the variance. Recently, Kothari and O'Donnell \cite{kothari2023mean} further improved existing works, which we utilize as a black box in our study.

In this study, we extend existing methods to estimate expected values for random variables with infinite variance but finite $(1+\delta)^{th}$ moment, where $\delta \in (0,1]$. Our algorithm exhibits a quadratic speed-up in relation to both the accuracy parameter $\epsilon$ and the $(1+\delta)^{th}$ moment of the distribution. It achieves $\epsilon$-Mean absolute error (MAE), a stronger metric than the `large probability guarantee' used by most existing quantum algorithms.

We provide both classical and quantum lower bounds, demonstrating that our quantum algorithm nearly quadratically surpasses all classical algorithms, not just traditional Monte Carlo methods, and is nearly optimal among all quantum algorithms with respect to both $\epsilon$ and the $(1+\delta)^{th}$ moment. Our lower bounds are established by reducing Grover-type search problems to our mean estimation problems.

\section{Quantum Monte Carlo algorithm}\label{sec:QMC}
The canonical problem in Monte Carlo is that of computing $\bE[X]$. Its current quantized versions
\cite{montanaro2015quantum, Hamoudi2021Quantum, kothari2023mean} focus on estimating the expectation of a random variable with finite variance. We show how to use similar ideas to these results (and in particular to \cite{montanaro2015quantum, kothari2023mean}), when the random variable has an infinite variance. We do need some hypothesis:  we will assume that there is a $0<\delta\le 1$ and a known $C$ such that $\bE[\lvert v(\calA)\rvert^{1+\delta}]:= C_{1+\delta}\leq C < \infty$. 

We aim to prove the following:
\begin{theorem}\label{thm:quantum-heavy}
    There is a procedure that, given any (classical or quantum) algorithm $\mathcal A$ which outputs a random variable $v(\mathcal A)$ satisfying $ \bE[\lvert v(\calA)\rvert^{1+\delta}] =  C_{1+\delta} < \infty$  for some $\delta \in (0,1]$,
    and an upper bound $C$ on $C_{1+\delta}$, constructs a quantum algorithm which
    estimates $\bE[v(\calA)]$ up to additive error $\epsilon$ with success probability at least $4/5$ by using $\calO\left(C^{1/2\delta} \epsilon^{-(1+\delta)/2\delta}\right)$ samples (i.e. applications of
    a quantum-adapted version of $\calA$, as in \cite{kothari2023mean}). The hidden factor 
    within $\calO$ depends on $\delta$. 
\end{theorem}

\begin{remark}
As in all similar studies, we primarily focus on how the sample complexity grows in relation to 
the desired additive error $\epsilon$. Here we treat
the sample complexity as a two-variate function of both $\epsilon$ and $C$.
This allows investigation of cases like $C=\epsilon$, but our formula is also valid when $C$ is very large (which is the more typical case). 
Further, our proof will imply (although we do not state it explicitly in the theorem), that when $C^{1/2\delta} \epsilon^{-(1+\delta)/2\delta} = o(1)$, i.e. when $C = o(\epsilon^{1+\delta})$,
the constant 0 estimates the additive error to accuracy $\epsilon$, and so the sample complexity is 0.
\end{remark}

The result above establishes a quadratic speedup when compared to the classical or `vanilla' Monte Carlo algorithm, which is optimal within a constant factor, and operates by simply outputting the average of the samples it generated.

\section{The Model}

Our paper aims to analyze the sample complexity of estimating the mean of a random variable $X$ in the quantum setting, when we have additional information about $X$, namely its $(1+\delta)^{\text{th}}$ moment. We build upon the work of \cite{montanaro2015quantum} and \cite{kothari2023mean}, and in fact, we use their results in a black box manner. Like in these articles, we must pay special attention to the model. In this section we clarify:
\begin{enumerate}
    \item How does our algorithm access $X$?
    \item How do we interpret the infinite range of $X$, since quantum computers are inherently finite?
\end{enumerate}
To answer the first question: We adopt the model presented in \cite{kothari2023mean}. Classical sampling theory is based on an oracle model, in which a "button" is pushed to obtain a random sample of the variable. We can think of any random variable of interest as $X(u)$, where $u \in [0,1]$. When the button is pushed, the oracle (or `automaton') generates a uniformly random $u \in [0,1]$ and provides us with the corresponding value $X(u)$. The complexity of an algorithm is measured by the number of times the button is pushed.

In practice, the oracle is implemented by an algorithm ${\cal A}$, which returns a value $v({\cal A}) \sim X$. When introducing quantum concepts into this model, we transform ${\cal A}$ into a quantum operator (using well-known methods), which can be applied on a superposition of samples. The quantum sampling complexity is now defined as the number of applications of this operator within the sampling algorithm, which is now a quantum algorithm.

\medskip

\noindent{\bf The Code of Kothari and O'Donnell \cite{kothari2023mean}.} We call
${\cal A}$ `the code,' which is either an input-less 
randomized circuit or a quantum circuit with input $|0^n\rangle$ and a final measurement of $k$ qubits
that returns the random variable $X = v({\cal A}) \in \mathbb{R}$.
While Montanaro also assumes the availability of the code for ${\calA}$ and ${\calA}^{-1}$,
he modifies it less than \cite{kothari2023mean}.  Kothari and O'Donnell post-process 
${\cal A}$ to turn it into the phase shift operator $e^{-2 i \arctan v({\cal A})}$.
With this definition Kothari and O'Donnell prove:

\begin{theorem}[\cite{kothari2023mean}]\label{thm:KO}
There exists a quantum sampling algorithm (with classical control) that, for any code ${\cal A}$ and any $n \in \mathbb{N}$, applies the Kothari-O'Donnell operator as a black box $\calO(n)$ times, and estimates the expected value $\bE[v(\calA)]$ with additive error $\sigma(v(\calA))/n$ and success probability at least $0.9$.
\end{theorem}

\subsection{Infinite Random Variables}

When we have a random variable $X$ that does not have a bounded support,
there are different 
ways to think about it
in an algorithmic framework.

\noindent{\bf An infinite ${\cal A}$.} There exist classical algorithms representing infinite random variables. Consider an experiment in which we perform consecutive coin flips. After the first instance of `heads' appears following $k$ occurrences of `tails,' we output $2^{\alpha k}$ (where $\alpha < 1$). The output, $v(\calA)$, possesses an infinite variance as long as $\alpha > 1/2$. However, it has a finite mean and a finite $(1+\delta)^{\rm th}$ moment for any $\delta < \frac{1}{\alpha} - 1$. In attempts to quantize such a model, we encounter challenges due to the variable-time nature of ${\cal A}$.

\noindent{\bf A family ${\cal A}_n$.} In this case we truncate 
$X$ by multiplying it with the indicator function of $X\le n$. The thus-obtained random variable, $X_n$ is finite, and we quantize ${\cal A}_n$ that implements 
$X_n$. We can study the parameters of this quantization as $n$ tends to infinity.

\noindent{\bf A finite ${\cal A}$ 
with a large
$\bE[\lvert v(\calA)\rvert^{2}]$
but a small
$\bE[\lvert v(\calA)\rvert^{1+\delta}]$.}
Here our attitude is that the expression 
`heavy tail' should be taken with a grain of salt, since in reality no 
random variable is infinite, 
and what in practice `heavy tail' really means is that $\bE[X^{2}]$ is very large, but $\bE[X^{1+\delta}]$ is relatively small. It is sufficient to say that we estimate the running time of the quantization in terms of 
the parameter $\bE[\lvert X\rvert ^{1+\delta}]$, regardless of $\bE[X^{2}]$.

The type of sampling problem we address in this article involves a given target additive error $\epsilon$ for the mean estimation problem. \emph{We must assume that the algorithm initially receives an upper bound $C$ on $\mathbb{E}[\lvert X\rvert ^{1+\delta}]$ as in Theorem \ref{thm:quantum-heavy}.} Theorem \ref{thm:quantum-heavy} could not 
hold with any finite bound on the sample size if 
$C$ were not given as part of the input. Our algorithm then proceeds by determining an appropriate cutoff value, depending on both $\epsilon$ and $\mathbb{E}[\lvert X\rvert^{1+\delta}]$. After the application of this cutoff, the truncated variable becomes finite, rendering the distinction between the three aforementioned methods of thinking about an infinite variable irrelevant.

\section{Background on moments}

The sample complexity of the vanilla Monte Carlo method with finite variance can be assessed using the classical Central Limit Theorem; see, for example, \cite{Asmussen2007StochasticSA}. The analog result to the standard Central Limit Theorem can also be obtained in the infinite variance case, assuming power-law-type distributions; see \cite{Sato1999LvyPA}. More generally, assuming only that $\calA$ has infinite variance
(but when $\bE[\lvert v(\calA)\rvert^{1+\delta}] < \infty$) the following result by Marcinkiewicz-Zygmund can be used to obtain a rate of convergence bound for the empirical mean of centered i.i.d. random variables. We let $p=1+\delta$.

\begin{theorem}[Marcinkiewicz-Zygmund generalized law of large numbers (\cite{gut2005probability}, pg 311)]\label{thm:MZ-LLN} 
Let $1 \leq p < 2$.
Suppose that $ X_1, X_2,...$ are i.i.d. random variables, and $S_n = \sum_{k=1}^n X_k$, $n \geq 1$. If  $\bE|X|^p < \infty$ and $\bE [X] = 0$, then
\[
\bE 
\left| 
\frac{S_n}{n^{1/p}}
\right|^p
=
\bE 
\frac{|S_n|^p}{n}
\stackrel{n \to \infty}{\longrightarrow}
0.
\]
\end{theorem}

We will first assume $C_{1+\delta} = 1$. From the above theorem, we know the average of i.i.d. copies of $v(\calA)_k$ satisfies:

\begin{align*}
    \bE\left[\left\lvert\frac{\sum_{k=1}^n v(\calA)_k - \bE[v(\calA)]}{n}\right\rvert\right] \leq \bE\left[\left\lvert\frac{\sum_{k=1}^n v(\calA)_k - \bE[v(\calA)]}{n}\right\rvert^{p}\right]^{1/p} = o(n^{(1-p)/p}).
\end{align*}
Therefore, one can choose $n = \calO(\epsilon^{-p/(p-1)}) = \calO(\epsilon^{-(1+\delta)/\delta})$ to ensure the `vanilla' Monte Carlo estimator has at most $\epsilon$-error. 

For any given \( C_{1+\delta} \), it's important to recognize that estimating \( \mathbb{E}[v(\mathcal{A})] \) with an error  of \( \epsilon \) is equivalent to estimating \( \mathbb{E}[v(\mathcal{A})/C] \) with an error of \( \epsilon/C \), where \( C > 0 \). Consequently, we can select \( C = C_{1+\delta}^{1/(1+\delta)} \), ensuring that \( \mathbb{E}[|v(\mathcal{A})/C|^{1+\delta}] \leq 1 \). Following this logic, the required sample size \( n \) can be derived as \( \mathcal{O}((\epsilon/C)^{-(1+\delta)/\delta}) \), which simplifies to \( \mathcal{O}\left(C_{1+\delta}^{1/\delta} \epsilon^{-(1+\delta)/\delta}\right) \).

\section{Upper Bounds}\label{sec:upper}

Before proving Theorem \ref{thm:quantum-heavy} let us introduce some definitions. Following the notations in \cite{montanaro2015quantum}, let $\calA_{<x}, \calA_{x,y}, \calA_{\geq y}$ be  algorithms 
defined by executing $\calA$ and
\begin{itemize}
  \item $\calA_{<x}$: If \(v(A) < x\), output \(v(A)\), otherwise output 0;
  \item $\calA_{x,y}$: If \(x \leq v(A) < y\), output \(v(A)\), otherwise output 0;
  \item $\calA_{\geq y}$: If \(y \leq v(A)\), output \(v(A)\), otherwise output 0.
\end{itemize}
For an event $A$ we will denote 
its characteristic function by $\mathbb I(A)$. This is a random variable that takes 1 on $A$
and 0 on $\overline{A}$. Clearly, $\bE(\mathbb I(A)^{\alpha}) = \bP(A)$
for any $\alpha>0$. Also, $v(\calA_{<x}) = v(\calA) \cdot \mathbb I(\calA < x)$, etc.

\begin{proof} (of Theorem \ref{thm:quantum-heavy}) Let $X$ be the random variable $v(\calA_{\ge 0})$. We demonstrate how to estimate $X$ with an additive error of $\epsilon/2$ and success probability at least 0.9. By doing so, we can obtain an $\epsilon$ estimation of $\bE(v(\calA))$ with success probability at least 0.8 by making the analogous estimate for $\bE(-v(\calA_{\le 0}))$ and subtracting this value from the first.
    
    Let $y \geq 0$. By H\"older's inequality, applied on 
    $v(\calA)$ and $\mathbb I(\calA \geq y)$ with norms $p = 1+\delta$ and $q = (1+\delta)/\delta$ respectively,
    we have
    \begin{align*}
    \bE[X\cdot \mathbb I(X \geq y)] &\leq \bE[ X^{1+\delta}]^{1/(1+\delta)} \,\cdot\, \bP[X \geq y]^{\delta/(1+\delta)}\\ 
    & \leq  \bE[ X^{1+\delta}]^{1/(1+\delta)} 
     \,\cdot\, 
    \left(\frac{\bE[ X^{1+\delta}]}{y^{1+\delta}}\right)^{\delta/(1+\delta)}\\ 
    & = C_{1+\delta} y^{-\delta}. 
    \end{align*}
Therefore, setting $y := \left(\frac{8C_{1+\delta}}{\epsilon}\right)^{1/\delta}$ yields 
\begin{equation}\label{eq:largecase}
\bE[v(\calA_{\ge y})] \; = \; \bE[X\cdot \mathbb I(\calA \geq y)] \; \leq \; 
 C_{1+\delta} y^{-\delta} =
C_{1+\delta} \frac{\epsilon}{8C_{1+\delta}}\; = \;
\frac{\epsilon}{8}
\end{equation}
Let $Z$ be the $\left(\frac{8C}{\epsilon}\right)^{1/\delta}$-truncated version of $X$
(recall, $C$ was an upper bound on $C_{1+\delta}$, given to the user in advance):
\[
Z = X\cdot \mathbb I\left(X \leq \left(\frac{8C}{\epsilon}\right)^{1/\delta}\right)
\]
Then due to $Z^{2}\le X^{1+\delta} \left(\frac{8C}{\epsilon}\right)^{\frac{1}{\delta}(1-\delta)}$
we have
\begin{align}
    \bE[ Z^{2}] & \le \; \bE[ X^{1+\delta}] \left(\frac{8C}{\epsilon}\right)^{\frac{1-\delta}{\delta}} \\
    & = \;  C_{1+\delta} \left(\frac{8C}{\epsilon}\right)^{-1} \cdot
    \left(\frac{8C}{\epsilon}\right)^{\frac{1}{\delta}} \; \le \; 
    \frac{1}{8} \epsilon\cdot \left(\frac{8C}{\epsilon}\right)^{\frac{1}{\delta}}
\end{align}
so
\[
\sigma\left(Z\right)^2/\epsilon^{2} \; \le \;
\bE[ Z^{2}]/\epsilon^2 \; = \;
\epsilon^{-2}\cdot \frac{1}{8} \epsilon\cdot \left(\frac{8C}{\epsilon}\right)^{\frac{1}{\delta}}
\; = \; 8^{\frac{1-\delta}{\delta}} \epsilon^{-\frac{1+\delta}{\delta}} C^{\frac{1}{\delta}}
\]
By Theorem \ref{thm:KO} if we set the sample size to be $n = \calO(\sigma\left(Z\right)/\epsilon)$, we can bound the additive error of approximating $Z$ by $\epsilon/8$ with probability at least 0.9. We have 
\[
\calO(\sigma\left(Z\right)/\epsilon) \; = \;
\calO\left(\sqrt{ 8^{\frac{1-\delta}{\delta}} \epsilon^{-\frac{1+\delta}{\delta}} C^{\frac{1}{\delta}} } 
\right) \; = \; \calO\left( \epsilon^{-\frac{1+\delta}{2\delta}} C^{\frac{1}{2\delta}} \right)
\]
The same estimator (without change) approximates $\bE[X]$ with additive error \\
$\epsilon/8 + \bE[X-Z] \le \epsilon/8 + \epsilon/8 = \epsilon/4$, what we needed. \end{proof}
 
 \subsection{Mean absolute error}\label{sec:MAE}
Theorem \ref{thm:quantum-heavy} designs an algorithm which estimates $\bE[v(\calA)]$ up to additive error $\epsilon$ with success probability $0.8$. The success probability can be  boosted to $1- \eta$ for any $\eta >0$ by repeating the algorithm in Theorem \ref{thm:quantum-heavy} for $\calO(\log(1/\eta))$ times and then take the median. This is  referred to as the `powering lemma', and  see Lemma 1 in \cite{montanaro2015quantum} for a proof. 

The algorithm presented in Theorem \ref{thm:quantum-heavy} (and those in \cite{montanaro2015quantum}) ensures that the results are accurate with a high probability. However, this doesn't straightforwardly convert to a guarantee on the mean absolute error (MAE) —a measure that is arguably more common, particularly in the analysis of classical algorithms. Additionally, the MAE is a slightly stronger metric because it implies  the `high probability' assurances through the immediate application of the Markov inequality. Our next result shows that with a minor modification to the algorithm in Theorem \ref{thm:quantum-heavy} and utilizing the `power lemma,' one can establish a guarantee on the MAE, only introducing an additional logarithmic factor dependent on $C$ and $\epsilon$.

\begin{theorem}\label{thm:quantum-heavy-expected-error}
    Fix accuracy level $\epsilon$. There is a quantum algorithm $\widetilde{\Qmc}(\epsilon)$
    that outputs an estimator $v\left(\widetilde{\Qmc}(\epsilon)\right)$ which estimates $\bE[v(\calA)]$ within  $\epsilon$-MAE, i.e.,
    \[
    \bE\left[\left\lvert v\left(\widetilde{\Qmc}(\epsilon)\right) - \bE[v(\calA)]\right\rvert\right] < \epsilon,
    \]
 
by using $\calO\left(\log\frac{C}{\epsilon} \cdot C^{1/2\delta} \epsilon^{-(1+\delta)/2\delta}\right) $ samples. The hidden factor within $\calO$ depends on $\delta$. 
\end{theorem}
\begin{proof}[Proof of Theorem \ref{thm:quantum-heavy-expected-error}]
    For $a < b$ define the clipping function $\Clip_{a,b}(x)$ as 
    \begin{align*}
        \Clip_{a,b}(x) := \begin{cases}
        a \qquad \text{If}~~ x < a\\
        x \qquad \text{If}~~ a\leq x \leq b \\
        b \qquad \text {If}~~ x > b. \\
        \end{cases}
    \end{align*}
Let us denote the algorithm in Theorem \ref{thm:quantum-heavy} that achieves an accuracy of $\epsilon$ and a success probability $p$, with $\Qmc(\epsilon, p)$, and its output with $v\left(\Qmc(\epsilon, p)\right)$. Consider a 'clipped version' of $\Qmc(\epsilon, p)$, that we will denote by $\Qmc^\Clip(\epsilon, p)$, and which we define as follows: after running $\Qmc(\epsilon, p)$ we apply the clipping function $\Clip(-C^{1/(1+\delta)}, C^{1/(1+\delta)})$ on $v\left(\Qmc(\epsilon, p)\right)$.  The above modification satisfies the following conditions: 
\begin{enumerate}
\item It preserves or reduces the estimation error, due to
$\lvert \bE[v(\calA)]\rvert \leq  \bE[\lvert v(\calA)\rvert] 
\leq C^{1/(1+\delta)}$, which is in turn a consequence 
of Jensen's inequality; 
\item It ensures that the maximum error does not exceed $2C^{1/(1+\delta)}$; 
\item It does not increase the computational cost, since clipping is a mere post-processing 
which does not require taking additional samples. 
\end{enumerate}

Corresponding to 1-3, the modified algorithm $\Qmc^\Clip(\epsilon/2, 0.8)$ produces an estimator that has an additive error of no more than $\epsilon/2$ with a probability of at least $0.8$. It has an upper bound of 
$2C^{\frac{1}{1+\delta}}$
on the maximum possible error, and it incurs a computational cost of $\calO\left(C^{\frac{1}{2\delta}}
\epsilon^{-\frac{1+\delta}{2\delta}}\right)$.  

Our final algorithm $\widetilde{\Qmc}(\epsilon)$ utilizes the powering lemma to boost the success probability to $1 - \epsilon/(4C^{1/(1+\delta)})$. This enhancement is achieved by running $\Qmc^\Clip(\epsilon/2, 0.8)$ for $\calO(\log(4C^{1/(1+\delta)}/\epsilon))$ times, and then taking the median of the results. This process ensures that the resulting estimator has an additive error of at most $\epsilon/2$ with a probability of at least $1 - \epsilon/(4C^{1/(1+\delta)})$, while the maximum error remains  at most $2C^{1/(1+\delta)}$. Consequently, the MAE can be readily upper bounded by:
\[
\bE\left[\left\lvert v\left(\widetilde{\Qmc}(\epsilon)\right) - \bE[v(\calA)]\right\rvert\right]  < \frac{\epsilon}{2} \cdot 1+  2C^{1/(1+\delta)} \frac{\epsilon}{4C^{1/(1+\delta)}} = \epsilon
\]
The sample complexity of $\widetilde{\Qmc}(\epsilon)$ is
$
\calO(\log(4C^{1/(1+\delta)}/\epsilon)) \cdot \calO\left(C^{1/2\delta} \epsilon^{-(1+\delta)/2\delta}\right).
$
\end{proof}

\subsection{Absolute \((1+\delta)^{\text{th}}\)-central moment}
We extend our analysis to include cases where the absolute \((1+\delta)^{\text{th}}\)-central moment of \( v(\calA) \) has a known upper bound, \(\Ccen\), that is, \(\mathbb{E}\left[\lvert v(\calA) - \mathbb{E}[v(\calA)]\rvert^{1+\delta}\right] \leq \Ccen\). In the $\delta = 1$ case, this corresponds to an upper bound of the variance which has been considered in \cite{montanaro2015quantum, kothari2023mean}. Following the idea
in \cite{montanaro2015quantum},  our Theorem \ref{thm:quantum-heavy} can also be adapted  to cover this case with cost $\calO\left(\Ccen^{1/2\delta} \epsilon^{-(1+\delta)/2\delta}\right)$ 
(substituting \(C_{1+\delta}\) with \(\Ccen\)).

Our revised algorithm begins by running \(\calA\) once, and storing the result as 
\( v_0 \). We then treat \( v_0 \) as a constant, and define \(\calA'\) 
such that \(v(\calA') = v(\calA) - v_0 \).  Subsequently, we run $\Qmc(\epsilon,0.9)$
of Theorem \ref{thm:quantum-heavy} (as defined in the previous section) with $C_{1+\delta} := 9 \cdot 2^{1+\delta}\Ccen$ on $\calA'$. Finally we add $v_0$ to the output of $\Qmc$ to obtain our final output. 

We claim our algorithm estimates $\bE[v(\calA)]$ up to additive error $\epsilon$ with success probability at least $0.8$.  This is because the probability of the event 
\[
E := \{\lvert v_0 - \mathbb{E}[v(\calA)]\rvert > (9\Ccen)^{1/(1+\delta)}\} 
\]
is at most $\frac{\Ccen}{9\Ccen} = \frac{1}{9}$.  For any \( v_0 \) within \( E \), the output of \(\calA'\) has a \((1+\delta)^{\text{th}}\) moment upper bounded by \(9 \cdot 2^{1+\delta}\Ccen\). Our final algorithm succeeds if $\neg{E}$ holds and $\Qmc$ succeeds, which has probability at least $(1-1/9) \times 0.9 = 0.8 $.

\section{Lower bounds}\label{sec:lower}

In this section we prove a lower bound of $\Omega\left(C^{1/2\delta} \epsilon^{-(1+\delta)/2\delta}\right)$ on the number of samples necessary for a quantum sampler, and an
$\Omega\left(C^{1/\delta} \epsilon^{-(1+\delta)/\delta}\right)$ for a classical sampler
(recall that $C$ was an upper estimate on $C_{1+\delta}$, given to the algorithm). 
Given that we have matching upper bounds (Theorems \ref{thm:quantum-heavy}, \ref{thm:MZ-LLN}) we now have a comprehensive understanding of the situation. 

We establish our lower bounds by reducing Grover-type search problems to our mean estimation problems. This connection is not new and is mentioned or utilized in one form or another at least in \cite{montanaro2015quantum, 
Belovs2019QuantumAF, Hamoudi2021Quantum, kothari2023mean}. However, we believe it is sensible to present the full argument not only for the sake of completeness and for didactic purposes but also to demonstrate how the bounds are obtaied across all ranges of
$\epsilon$ and $C_{1+\delta}$ in both classical and quantum contexts, aside from constant factors.

Let $[N] = \{1,\ldots,N\}$, $M\in \mathbb{R}\setminus\{0\}$. 
We define a set of functions of the type $[N]\rightarrow \mathbb{R}$, altogether $N$ of them, by:
\begin{equation}
    f_{i}(k) = 
    \left\{\begin{array}{lll}
    0 & {\rm if} & k \neq i \\
    M & {\rm if} & k =i \\
    \end{array}\right. \;\;\;\;\;\;\;\;\;\; 1\le i, k \le N
\end{equation}
We define an additional function, $f_{0}$ which takes $0$ everywhere in $[N]$. 
Let
\begin{equation}
    \calP = \{f_{1},\ldots,f_{N}\} \;\;\;\;\;\; \calN = \{f_{0}\}
\end{equation}

Define the function classification problem, $\Pi$, as follows:

\noindent{\bf Problem $\Pi(N,M)$:} Alice (the oracle) hides a function $f \in \mathcal{P} \cup \mathcal{N}$. Our task is to answer `yes' if $f \in \mathcal{P}$ and `no' if $f \in \mathcal{N}$. (Here the actual value of $M$ plays no role, and it is only important that $M \neq 0$.) The oracle is accessed 
in the manner of the usual query model, which is either randomized or quantum.

The task of detecting a target in an unsorted $N$-item database, which needs $\Omega(N)$ queries in a classical setting (folklore, but see \cite{Bennett1981Relative}) and $\Omega(\sqrt{N})$ in a quantum one (see \cite{Bennett1997Strengths}), is precisely our classification problem.
 \begin{theorem}\label{thm:lbds}\cite{Bennett1981Relative, Bennett1997Strengths}
     The complexity of solving $\Pi(N,M)$ with at most 0.2 probability of error on {\em all} inputs
     in $\calP\cup\calN$ has query complexity
     \begin{enumerate}
         \item $\Omega(N)$ in the randomized query model
         \item $\Omega(\sqrt{N})$ in the quantum query model.
     \end{enumerate}
 \end{theorem}

To carry out our plan, we also define the following mean-estimation problem, $\Sigma$:

\noindent{\bf Problem $\Sigma(N,M)$:} If $f$ is any of $f_{0},\ldots,f_{N}$,
     we define $\mathcal{A}_{f}$ to be the sampling oracle for the random variable that draws 
     $k \in [N]$ randomly and uniformly and returns $f(k)$. Design an estimator 
     that estimates $\mathbb{E}[v(\mathcal{A}_{f})]$ with additive error less or equal than
     $\frac{M}{3N}$ for any $f \in \mathcal{P} \cup \mathcal{N}$. 

We will reduce problem $\Pi$ to $\Sigma$. Using Theorem \ref{thm:lbds} we then deduce lower bounds for $\Sigma$.
     However, before we do this, we should address a significant omission. As noted in \cite{kothari2023mean}, theoretically, no sampling is necessary if we have access to the 'source code' of a random variable $X$. To ensure that our lower bounds are meaningful, it must also be assumed that the sampling oracle 
     for any $\mathcal{A}_{f}$ in $\Sigma$
     is derived from the code in a specific, restricted manner. In the classical context, the nature of access to the code is straightforward — any alterations are not permitted. However, in the quantum scenario, some minor form of alteration is inevitable, especially when quantizing if the code was originally classical.

    In our proof of the lower bound, we adopt the model from \cite{kothari2023mean} and assume three operators through which a random variable $X:\Omega\rightarrow\mathbb{R}$ is accessed (sampled). First, we require a unitary circuit, which we call the \emph{Synthesizer}, that represents a probability distribution $\Omega\rightarrow\mathbb{R}_{\ge 0}$ as a unitary map. This map transforms $|\vec{0}\rangle$ into $\sum_{\omega\in\Omega}\sqrt{p(\omega)} |\omega\rangle |\text{garbage}_{\omega}\rangle$. Here, $p(\omega)$ denotes the probability of $\omega$. We also need access to operators controlled-$\mathcal{X}$ and controlled-$\mathcal{X}^{\dagger}$, where $\mathcal{X}$ is any unitary circuit that behaves as follows:
    \[
    \mathcal{X} |\omega\rangle |0^b\rangle |0^c\rangle = |\omega\rangle |X(\omega)\rangle |0^c\rangle
    \]

    \emph{What is relevant to our lower bound is that Theorem \ref{thm:lbds} provides a lower bound on the number of applications of $\mathcal{X}$ and $\mathcal{X}^{\dagger}$ in any query algorithm that classifies $X$, i.e. tells whether it is all-zero or takes a non-zero value for a specific $\omega$. Additionally, it is important that the probability distributions behind all $\mathcal{A}_{f}$ are identical, thus the same \emph{Synthesizer} is used. Consequently, the complexity ultimately hinges on the 
    number of uses of $\mathcal{X}$ and $\mathcal{X}^{\dagger}$.}

    It is in the above sense that in the following lemma we talk about the sample complexity of 
    $\mathcal{A}_{f}$.

 \begin{lemma}\label{lem:simul}
     Assume that there is a classical sampling procedure, $\Cmc(\frac{M}{3N}, 0.8)$ (resp. quantum procedure 
     $\Qmc(\frac{M}{3N}, 0.8)$),
     that solves $\Sigma(N,M)$ with sample complexity $n$. Here parameters 
     $\frac{M}{3N}$ and $0.8$ should be read as 'the additive error of the estimator is 
     $\le \frac{M}{3N}$ with probability at least 0.8', for any $v(\calA_{f})$, 
     $f\in \calP\cup\calN$.
     Then such a sampler 
     can be turned into a classical (resp. quantum) query algorithm
     that solves $\Pi(N,M)$ with at most 0.2 probability of error for any input $f\in \calP\cup\calN$, 
     and which has complexity $n$.
 \end{lemma}

 \begin{proof} By the definition of $\calA_{f}$ (recall $\calA_{f}$ from
 the definition of $\Pi(N,M)$) we have:
     \begin{align}\label{eq:ncase}
     Case\; f\in\calN: \;\;\;\;\;\bE[v(\calA_{f_{0}})] & = 0 \\\label{eq:pcase}
     Case\; f\in\calP: \;\;\;\;\;\bE[v(\calA_{f_{i}})] & = \frac{M}{N} \;\;\;\;\;\;\;\;\;\; {\rm for}  \; 1\le i \le N 
     \end{align}
     Thus, estimating $\bE[v(\calA_{f})]$ to within an additive error 
     of $\frac{M}{3N}$ differentiates between cases 
     (\ref{eq:ncase}) and (\ref{eq:pcase}).
     The above fact establishes the connection between $\Sigma(N,M)$ and $\Pi(N,M)$. To convert the query complexity lower bound for $\Pi(N,M)$ into a sample size lower bound for
     $\Sigma(N,M)$ , we argue that our sampling algorithm is in fact a special query algorithm, where the samples become the queries. The sampling model is more restrictive than the query model because the samples, informally speaking, 'come from nature'—that is, they are randomly chosen, rather than being selected through 'clever' algorithmic choices. Due to this, the sampling model is contained within the query model, allowing us to apply any existing lower bound from the more powerful query model to establish a bound in the less powerful sampling model.

In the final step, we turn the output of $\Cmc(\frac{M}{3N}, 0.8)$ (resp. $\Qmc(\frac{M}{3N}, 0.8)$, in the quantum case)
to the output of $\Sigma(M,N)$. We output 'no' if our mean-estimator returns a value less than $\frac{M}{2N}$, and 'yes' whenever it returns a value at least $\frac{M}{2N}$.

     Assume that $f\in \calN$, i.e. that $f=f_{0}$. Then $\bE[v(\calA_{f})] = 0$
     and the probability that $\Cmc(\frac{M}{3N}, 0.8)$ (resp. $\Qmc(\frac{M}{3N}, 0.8)$)
     returns an estimation greater than
     $\frac{M}{3N}$ (and in particular greater or equal than $\frac{M}{2N}$)
     is at most 0.2. Similarly one can see that if $f\in \calP$ (in 
     which case $\bE[v(\calA^{r}_{f})] = \frac{M}{N}$)
     the probability to get an estimated value below $\frac{2M}{3N}$ is at most 0.2.

     In conclusion, our mean estimator turned query algorithm solves the classification 
     problem for all $f\in \calP \cup \calN$ with the upper bound on the probability of failure 
     on any input, required, and with query complexity equaling to the sample size of $\Cmc(\frac{M}{3N}, 0.8)$ (resp. $\Qmc(\frac{M}{3N}, 0.8)$, in the quantum case). 
 \end{proof}

 Lemma \ref{lem:simul} and Theorem \ref{thm:lbds} establish a lower bound on the number of samples required by $\Cmc(\frac{M}{3N}, 0.8)$ (or $\Qmc(\frac{M}{3N}, 0.8)$ in the quantum case), namely $\Theta(N)$ and $\Theta(\sqrt{N})$ respectively. 
To finish our argument, we are going to calculate these bounds in terms of 
\begin{align}\label{eq:epseq}
\epsilon\; & = \; \frac{M}{3N} \\\label{eq:ceq}
C \; & = \; \max_{0\le i \le N} C_{1+\delta}[f_{i}] \; = \; 
\underbrace{C_{1+\delta}[f_{k}]}_{k=1,\ldots,N} = \frac{N-1}{N}\times 0
+ \frac{1}{N} \times M^{1+\delta} \; = \;
\frac{M^{1+\delta}}{N}
\end{align}

\begin{remark}
    Here we assume that what algorithms $\Cmc(\frac{M}{3N}, 0.8)$ and $\Qmc(\frac{M}{3N}, 0.8)$ receive as an upper bound for
    $C_{1+\delta}(f)$ is
    $C=\frac{M^{1+\delta}}{N}$ for any $f\in \calP \cup \calN$.
\end{remark}

First we express $M$ and $N$ in terms of $\epsilon$
and $C$. We have $N= \frac{M}{3\epsilon}$ and $N = \frac{M^{1+\delta}}{C}$, so
$\frac{1}{3\epsilon} = \frac{M^{\delta}}{C}$, which in turn gives 
\[
M = \left( \frac{C}{3\epsilon} \right)^{\frac{1}{\delta}} \;\;\;\;\;\; {\rm and}
\;\;\;\;\;\;
N = \left( \frac{C}{3\epsilon} \right)^{\frac{1}{\delta}} \frac{1}{3\epsilon}
= 3^{-\frac{1}{\delta}-1} C^{\frac{1}{\delta}} \epsilon^{- \frac{1+\delta}{\delta}}
\]
In the classical case Theorem \ref{thm:lbds} yields:
\[
    n \; = \; \Omega(N) \; = \; \Omega\left( C^{\frac{1}{\delta}} \epsilon^{- \frac{1+\delta}{\delta}} \right) 
\]
In the quantum case Theorem \ref{thm:lbds} yields:
\[
    n \; = \; \Omega(\sqrt{N}) \; = \; \Omega\left( C^{\frac{1}{2\delta}} \epsilon^{- \frac{1+\delta}{2\delta}} \right) 
\]
We summarize these results in the following theorem:
\begin{theorem}\label{thm:lb}
    Let $\delta >0$, and 
    $C$ be an upper bound on
    $C_{1+\delta}[v(\calA)] = \bE[|v(\calA)^{1+\delta}|]$. Then 
    if a classical sampler is able to give an additive $\epsilon$-approximation of the mean
    of $v(\calA)$ with confidence at least 0.8
    for {\em all} $\calA$s with $C_{1+\delta}[v(\calA)]\le C$, 
    then this procedure must have (worst case) sample complexity at least 
    $\Omega\left( C^{\frac{1}{\delta}} \epsilon^{- \frac{1+\delta}{\delta}} \right)$.
    The analogous lower bound for the quantum case is
    $\Omega\left( C^{\frac{1}{2\delta}} \epsilon^{- \frac{1+\delta}{2\delta}} \right)$.
\end{theorem}

\begin{remark}
    The classical portion of Theorem \ref{thm:lb} is also proven in Theorem 3.1 of \cite{devroye2016sub} in a different way. 
\end{remark}

\section{Towards a $\delta$-independent algorithm; Takeaway}

Our bounds have the following puzzling feature:

\emph{The lower bound argument in Section \ref{sec:lower} relies on a two-point construction that is independent of $\delta$. The $\epsilon$ associated with this particular construction is determined solely by its two parameters ($M$ and $N$), also independent of $\delta$. In other words, the extremal construction we give in the lower bound proof has nothing to do with $\delta$. The dependence on $\delta$ only emerges in Equations (\ref{eq:epseq}) and (\ref{eq:ceq}), when we express the above bound in terms of $C_{1+\delta}$. It turns out (see Equations (\ref{eq:crucial1}-\ref{eq:crucial3})) that the same construction yields bottleneck results for a variety of $\delta, \epsilon$ combinations. This naturally raises the question of how much the bound on $C_{1+\delta}$ is actually relevant to optimal algorithm design.}

The answer to this puzzle lies in a deeper understanding of what our result means.
We do not invent profoundly new algorithms, rather we show how to calculate the optimum setting 
of a threshold value, based on moment information about a random 
variable. It just happens so
that when the extra information given about $X$ is
its $(1+\delta)$-moment $0< \delta < 1$, this is
{\em sometimes} a better information than if we are given
the second moment. Namely, when the variable is sufficiently heavy tailed.

{\em Our algorithm makes a single decision based on $C\ge C_{1+\delta}$, namely it sets the 
upper cutoff to be $y = \left(\frac{8C}{\epsilon}\right)^{\frac{1}{\delta}}$.
After that it just runs \cite{kothari2023mean} on the truncated variable.}

To illuminate what this accomplishes, and when this is a good thing, 
consider the example of the two-point distribution. Let $M,N>1$ be integers, and define $\calA$ via:
\[
\bP[v(\calA)=0] = 1 - \frac{1}{N}   \;\;\;\;\;\;\;\;\;\;\; \bP[v(\calA)=M] = \frac{1}{N}
\]
Let us also set $\epsilon = \frac{M}{N}$. This value for $\epsilon$ represents an important threshold: when $\epsilon$ falls below the $\Theta(M/N)$, it is when the small probability event $v(\calA)=M$ starts to be significant when $\epsilon$-approximating the mean. For simplicity we set $C=C_{1+\delta}$. We have
\begin{align}\label{eq:crucial1}
C_{1+\delta} & = \frac{M^{1+\delta}}{N} \\\label{eq:crucial2}
{\rm cutoff \; value} \; = \; y \; = \; \left(\frac{8C_{1+\delta}}{\epsilon}\right)^{\frac{1}{\delta}} & = \Theta(M) \\\label{eq:crucial3}
{\rm number \; of \; samples} \; = \; t \; = \; 
\Theta\left(C_{1+\delta}^{1/2\delta} \epsilon^{-(1+\delta)/2\delta}\right)
& = \Theta(\sqrt{N})
\end{align}
Observe that the parameters that are crucial to how we run the algorithm, namely 
$y$ and $t$, do not contain $\delta$. Also, since the number of samples is $\Theta(\sqrt{N})$, 
we see that $\delta$ does not play a role in our running time analysis
either. 

Let us now consider the general case, where 
for some arbitrary $L>0$ we express $\epsilon$ as
$\epsilon = \frac{M}{LN}$. 
The case $L=1$ was the case above. Below we are going to discuss the 
$L> 1$ and the $L<1$ cases. The general formulas are:
\begin{align}
y  & = \Theta(L^{\frac{1}{\delta}}M) \\
t & = \Theta\left(L^{\frac{1+\delta}{2\delta}}\sqrt{N}\right)
\end{align}
Here the value of $\delta$ {\em does} play a role in setting the parameters of the corresponding algorithm. The sample size lower bound for
the combination $(C_{1+\delta}, \epsilon) = \left(\frac{M^{1+\delta}}{N}, \frac{M}{LN}
\right)$ in Section \ref{sec:lower} perhaps surprisingly, relies not on $M$ and $N$, but on the two-point distribution with parameters
$M' = L^{\frac{1}{\delta}}M$, $N' = L^{\frac{1+\delta}{\delta}}N$.
This hints at a potential improvement, {\em but only of course
if we were to be given further/different information about $\calA$ in advance.}

Indeed, let us first look at the $L>1$ case, and let us see what kind of prior information 
about $\calA$ would yield a better sample complexity.
We first compute $C_{2}=
\Theta\left(\frac{M^2}{N}\right)$. Assuming 
$C_{2}$ is the only piece 
of information we know about $\calA$ (as opposed to $C_{1+\delta}$), using the estimate of 
\cite{montanaro2015quantum}, we find that the sample size should be set as:
\[
\Theta\left(\frac{\sqrt{C_{2}}}{\epsilon}\right)
= \Theta\left(
\frac{M}{\sqrt{N}} \frac{LN}{M}
\right)
= \Theta\left(L\sqrt{N}
\right)
\]
Thus, if $L>>1$, at least in the case of the above two-point distribution, we are at a disadvantage when our algorithm is designed based on knowledge of $C_{1+\delta}$ rather than $C_{2}$:
\begin{equation}\label{eq:compare}
\underbrace{\Theta\left(L^{\frac{1+\delta}{2\delta}}\sqrt{N}\right)}_{\rm sample\; size, \; when \; C_{1+\delta} \; is \; known} \;\; {\rm vs.} \;\; \underbrace{\Theta\left(L\sqrt{N}\right)}_{\rm sample\; size, \; when\; C_{2} \; is \; known }
\end{equation}
Let us now consider the remaining scenario, where $L<1$. In this case, one should simply output 0 sampling nothing whatsoever
(since the mean is smaller than the error-threshold, $\epsilon$), the algorithm relying on $C_{1+\delta}$ is again sub-optimal. However, there is a catch: Without having any prior information about $\calA$, we do not know that we are in the good situation of not having to sample. 

Let us compare again two cases of prior information: 
when only the second moment ($C_{2}$) and when only $C_{1+\delta}$ is 
available about $\calA$. Looking at Equation (\ref{eq:compare}), we see that in the $L<1$ case basing 
our algorithm on knowing $C_{1+\delta}$ (rather than $C_{2}$) yields the smaller sample size. This is the only 
region for $L$, where setting the sample size according to our information about $C_{1+\delta}$ rather than $C_{2}$ is more advantageous. {\em The explanation is that having a bound $C_{1+\delta}$ tells our algorithm in a stronger sense than a bound on $C_{2}$ that there is nothing to look for, and so the algorithm can stop the search after fewer trials.}

Let $y\in \mathbb{R}$ be 
such that $\bE[v(\calA_{\ge y})]\le \epsilon/8$, and assume that $v(\calA)\ge 0$.
Although general distributions are not two-point, we can 
approximate $v(\calA)$ with a sum of $\log_{2}\frac{8y}{\epsilon}$
two-point distributions.
Subdivide the $[\frac{\epsilon}{8},y]$
interval into 
sub-intervals of exponentially increasing length. More precisely, let
$x_{0}=\frac{\epsilon}{8}$, $x_{i+1} = 2 x_{i}$ for $0\le i < K$, where $y< x_{K}\le 2y$. 
Let $q_{i} = \prob{\mathbb I(x_{i-1} \le \calA < x_i)}$.
The random variable $\calA_{x_{i-1},x_{i}}$ 
is 0 outside $\mathbb I(x_{i-1} \le \calA < x_i)$ and it
takes values between $x_{i-1}$ and $x_{i} = 2x_{i-1}$, so it is like a two-point distribution with 
probabilities $q_i$ and $1-q_{i}$. Writing these parameters in terms of the
$M$, $N$ parameters we have been using throughout this section, we have that 
$\calA_{x_{i-1},x_{i}}$ is a constant factor approximation of the two-point 
distribution where $N_{i} = 1/q_{i}$, $M_{i} = x_{i}$. We can also write 
$L_{i} = \frac{M_{i}}{\epsilon N_{i}}$. For a general random 
variable, using the above notation, we can informally say that 
our algorithm performs well based on the lower moment 
information when many $L_{i}$s are much smaller than 1.
This intuitive statement does not make any reference to $\delta$. Therefore, we may conjecture that knowledge of any $\delta$ moment can be replaced by the knowledge of $y, q_{1}, \ldots, q_{K}$. We capture this idea in the following propositions:

\begin{prop}\label{prop:nodelta}
Let $v(\calA)$ be a code to a random variable
that takes only non-negative values. Let $y\in \mathbb{R}$ be 
such that $\bE[v(\calA_{\ge y})]\le \epsilon/8$.
Let $x_{0}=\frac{\epsilon}{8}$, $x_{i+1} = 2 x_{i}$ for $0\le i < K$, where $y< x_{K}\le 2y$. 
Let
\[
q_{i} = \bP[v(\calA_{x_{i},x_{i+1}}) > 0]
\]
Assume that the sequence $q_{1},\ldots,q_{K}$ and $y$ are given to the user. 
Let $\nu_{i}$ indicate if $q_{i}x_{i}$ is greater than $\frac{\epsilon}{100K}$, i.e.
$\nu_{i} = 1$ if $q_{i}x_{i} >
\frac{\epsilon}{100K}$, otherwise it is zero.
Then there is quantum Monte Carlo sampler
with sample size 
\[
\tilde \Theta\left( \frac{1}{\epsilon}\sum_{i=1}^{K} \nu_{i} \, x_{i} q_{i}^{\frac{1}{2}}  \right)
\]
Here the tilde hides logarithmic factors in $y$ and $1/\epsilon$.
\end{prop}
\begin{proof}
First notice that if the user knows $q_{1},\ldots,q_{K}$ and $y$, then they also know $K$,
$x_{i}$ ($0\le i\le K$) and $\nu_{i}$ ($1\le i\le K$).
The second thing to notice is that for those intervals $[x_{i-1},x_{i}]$ 
with $\nu_{i}=0$
the expectation $\bE[(v(\calA_{x_{i-1},x_{i}})]$
represents an irrelevantly small part of $\bE[(v(v(\calA)]$, even when 
we sum these up for all such $i$.
 Finally, for $i$'s with $\nu_{i}=1$ one can sufficiently accurately estimate the contribution of
 $\bE[v(\calA_{x_{i-1},x_{i}})]$, using \cite{kothari2023mean}.
 To compute the sample complexity associated with this algorithm 
 note that $\bE[(v(\calA_{x_{i-1},x_{i}}))^2] = \calO(x_{i}^2 q_{i})$, so
 \[
 \epsilon^{-1} \sqrt{\bE[(v(\calA_{x_{i-1},x_{i}}))^2]} = \epsilon^{-1} x_{i}\sqrt{q_{i}}
 \]
 but to improve on the 
 accuracy of the estimate (we need accuracy $\frac{\epsilon}{100K}$),
 and on its success probability (we need success probability $1-\frac{1}{100K}$)
 we loose logarithmic factors in $y$ and $1/\epsilon$ 
 (in particular, $K = \lceil \log_{2} \frac{8y}{\epsilon} \rceil$).
\end{proof}

\begin{prop}\label{prop:best}
The sample size given in Proposition \ref{prop:nodelta} is at least as good,
aside from logarithmic factors, as of any algorithm designed only with 
the knowledge of $C_{1+\delta}$ alone. Here we assume, 
$y = \left(\frac{8C_{1+\delta}}{\epsilon}\right)^{\frac{1}{\delta}}$,
but it should be noted that generally choosing a value for $y$ is not problematic, as the running time depends only logarithmically on it. 
\end{prop}

This proposition was implicitly proven in this section, since every $\calA_{x_{i-1},x_{i}}$ ($1\le i\le K$)
can be treated as a two-point distribution, which we have fully analysed.
The fact that we neglected the terms $\bE[v(\calA_{x_{i-1},x_{i}})]$ where $\nu_{i}=0$
ensures that we never run into the $L_i << 1$ case. In fact, when information 
on all moments (including the zeroth) are given, the following can be can be proven:

\begin{prop}\label{prop:verybest}
Whenever $M\ge 1$ the following is true. Assume the task is to $\epsilon$-approximate a random variable.
Then, the sample size given in Proposition \ref{prop:nodelta} is at least as good,
aside from logarithmic factors, as of any algorithm that
prior to the running receives a function $C(\delta)$ defined on $0\le \delta\le 1$
with the property that $C(\delta)$ is an upper estimate on $C_{1+\delta}$, and 
$C(\delta)> 3 \epsilon M^{\delta}$.
\end{prop}

Despite that Proposition \ref{prop:nodelta} always gives better sample size than
Theorem \ref{thm:quantum-heavy}, calculating the dependence of the sample size 
with respect to the 
knowledge of $(1+\delta)^{th}$ moment alone has merits. 
First, it is not always realistic to expect that all the numbers $q_{1},\ldots,q_{K}$ are known to the user, 
but knowledge of some moment is often a typical scenario. 
The results in this section also shows that 
knowledge of different moments often do not aggregate; we should just 
pick the information that gives the smallest cutoff in Theorem \ref{thm:quantum-heavy}.

The main takeaway is that bounding the sample complexity 
relies on our prior knowledge of the variable, which can affect 
both the analysis and the algorithm itself.

\subsection{The Pareto distribution -- a case study}

The Pareto distribution \cite{johnson1994continuous}
is a continuous probability distribution often used to model heavy-tailed data, such as wealth or income distributions. It is characterized by two parameters, $x_{\text{min}} > 0$ and $\alpha > 0$, and is defined by the following probability density function (PDF):

\[
f(x; x_{\text{min}}, \alpha) = \begin{cases}
\frac{\alpha x_{\text{min}}^{\alpha}}{x^{\alpha+1}} & \text{for } x \geq x_{\text{min}}, \\
0 & \text{otherwise}.
\end{cases}
\]

The cumulative distribution function (CDF) for the Pareto distribution is given by:

\[
F(x; x_{\text{min}}, \alpha) = \begin{cases}
1 - \left(\frac{x_{\text{min}}}{x}\right)^{\alpha} & \text{for } x \geq x_{\text{min}}, \\
0 & \text{otherwise}.
\end{cases}
\]

The moments of the Pareto distribution can be calculated using the following formula for the $p = 1+\delta$-th moment:

\[
C_p = \int_{x_{\text{min}}}^{\infty} x^p f(x; x_{\text{min}}, \alpha) dx = \frac{\alpha x_{\text{min}}^p}{\alpha - p}
\]

For the above to be finite, the shape parameter $\alpha$ must be greater than $p$. 
To define our distribution, $v(\calA)$, we choose
\[
x_{\text{min}} := 1, \;\;\;\;\;\;\;\;\;  1 \;\;\; < \;\;\; 1 + \delta \;\;\; = \;\;\; p \;\;\; < \;\;\; \alpha
\;\;\; \le \;\;\; 2
\]
With these choices the second moment is infinite and
\begin{align}
\bE[v(\calA)] & = C_1 = \frac{\alpha}{\alpha - 1} \\
C_{1+\delta} & = \frac{1+\delta}{\delta} \; = \; \Theta(1);
\;\;\;\;\;\mbox{We will assume} \; C = C_{1+\delta}
\end{align}
We are going to investigate how our algorithm performs in terms of the sample size
as a function of $\epsilon$. Let us first calculate the cutoff, both the theoretically 
optimal and what our algorithm gives.

\medskip

\begin{tabular}{llc}
  The smallest $y_{0}$ such that $\bE[v(\calA_{\ge y_{0}})] 
  \le \epsilon$:   &  & $\Theta\left(\epsilon^{\frac{-1}{\alpha-1}}\right)$ \\
  & & 
  (Due to $\int_{y_{0}}^{\infty} x \frac{\alpha}{x^{\alpha+1}} dx
  = \Theta(y_0^{1-\alpha}) = \epsilon$.)\\[6pt]
  \multicolumn{3}{l}{But this is not the threshold our algorithm finds.} \\[6pt]
  The $y$ (cutoff) value that our algorithm finds:   &  & $\Theta\left(\left(\frac{C_{1+\delta}}{\epsilon}\right)^{\frac{1}{\delta}}\right) = \Theta(\epsilon^{-\frac{1}{\delta}})$
\end{tabular}

\medskip

We now start the preparation of comparing the sample complexity of three algorithms.

\medskip

\begin{center}
\begin{tabular}{llp{4.5in}}
$\Qmc(\epsilon, 0.8)$  &  & The algorithm in Theorem \ref{thm:quantum-heavy} \\
$\calA_{\rm ideal}(\epsilon)$  &  & The algorithm in Proposition \ref{prop:nodelta} \\
$\calA_{\rm KO}(\epsilon, \kappa)$  &  & The algorithm of Kothari and O'Donnell in \cite{kothari2023mean}, when run on $\calA_{0,\kappa}$. (For any fair comparison $\kappa$ must be at least $\epsilon^{\frac{-1}{\alpha-1}}$,
because that's where $\bE[v(\calA_{\ge\kappa})]<O(\epsilon)$.) \\
\end{tabular}
\end{center}

\medskip

The sample complexity of $\Qmc(\epsilon, 0.8)$, given by Theorem \ref{thm:quantum-heavy} is
\[
\calO\left(C^{1/2\delta} \epsilon^{-(1+\delta)/2\delta}\right) 
= \calO\left( \epsilon^{-(1+\delta)/2\delta}\right) 
\]
To calculate
the sample complexity of $\calA_{\rm ideal}(\epsilon)$, we first
compute $q_{1},\ldots,q_{K}$ as they are defined in Proposition \ref{prop:nodelta}. Let
$i_{0}= \left\lfloor \log_{2} \frac{8}{\epsilon} \right\rfloor$. Then
\begin{align}
q_{1} & = \ldots = q_{i_{0}-1}  = 0 \\
0 & \le q_{i_{0}} = \calO(1) \\
q_{i} & = \; \int_{x_{i}/2}^{x_{i}}\frac{\alpha \, dx}{x^{\alpha+1}}
\; = \; \Theta(x_{i}^{-\alpha}) \;\;\;\;\;\; {\rm for} \;\;\;\;\;\; i_{0} + 1 \le i \le K
\end{align}
To compute $\nu_i$ of Proposition \ref{prop:nodelta} we need $q_{i}x_{i}$.
When $i< i_{0}$ then $q_{i}x_{i}=0$. 
For $i = i_{0}$ we have $q_{i}x_{i}=0$.
Finally, when $i>i_{0}$ then
\[
q_{i}x_{i} \; = \; \Theta(x_{i}^{-\alpha}) \cdot x_{i} \; = \; \Theta(x_{i}^{1-\alpha})
\]
Utilizing ideas from the previous section, we understand that when $q_{i}x_{i} < \tilde\Theta(\epsilon)$, this is when the interval $[x_{i-1}, x_{i}]$ ceases to be significant. 
This happens when $x_{i} > \tilde\Theta(\epsilon^{\frac{-1}{\alpha-1}})$
(not surprisingly, since $\epsilon^{\frac{-1}{\alpha-1}}$ is the optimal cutoff value too).
The value of $\nu_i$ is also 0 for $i < i_{0}$, and possibly for $i = i_{0}$. 
Thus the sample complexity of $\calA_{\rm ideal}(\epsilon)$ is
\[
\tilde \calO\left(\frac{1}{\epsilon}
\sum_{i=i_{0}+1}^K  x_{i} q_{i}^{\frac{1}{2}} \right) \; = \;
\tilde \calO\left(\frac{1}{\epsilon}
\sum_{i=i_{0}+1}^K  x_{i}^{1-\frac{\alpha}{2}} \right)
\; = \; \tilde\Theta\left(\frac{1}{\epsilon}\epsilon^{\frac{-1}{\alpha-1}\frac{2-\alpha}{2}}\right)
\; = \; \tilde\Theta\left(\epsilon^{\frac{-\alpha}{2\alpha-2}}\right)
\]

Finally, the sample complexity of
$\calA_{\rm KO}(\epsilon, \kappa)$ 
is $\epsilon^{-1} \sqrt{\bE[|v(\calA_{0,\kappa})|^{2}] }$.
We calculate, that
\[
\bE[|v(\calA_{0,\kappa})|^{2}] =  \int_{1}^{\kappa} x^2 \frac{\alpha}{x^{\alpha+1}} dx
  = \Theta(\kappa^{2-\alpha})
\]
  and so the sample complexity becomes
  $\epsilon^{-1} \sqrt{\kappa^{2-\alpha}}$. We summarize:

\begin{center}
\begin{tabular}{l|c}
   Algorithm  &  Sample complexity \\\hline
 $\Qmc(\epsilon, 0.8)$    &   $ \Theta\left(\epsilon^{-\frac{1+\delta}{2\delta}}\right)$  \\[4pt]
 $\calA_{\rm ideal}(\epsilon)$ & $\tilde\Theta\left( \epsilon^{\frac{-\alpha}{2\alpha-2}}\right)$ \\[4pt]
 $\calA_{\rm KO}(\epsilon, \epsilon^{\frac{-1}{\alpha-1}})$ & 
 $ \Theta\left( \epsilon^{-\frac{\alpha}{2(\alpha-1)}}\right)$ \\[4pt]
 $\calA_{\rm KO}(\epsilon, \epsilon^{-\frac{1}{\delta}})$ & 
 $ \Theta\left( \epsilon^{-\frac{2-\alpha+2\delta}{2\delta}}\right)$ \\
\end{tabular}
\end{center}

\medskip

\noindent{\bf Conclusion:} 
Let us disregard logarithmic and constant factors. We first observe that $\mathcal{A}_{\mathrm{ideal}}(\epsilon)$ has the same sample comlexity as $\mathcal{A}_{\mathrm{KO}}(\epsilon, \epsilon^{\frac{-1}{\alpha-1}})$ ($\mathcal{A}_{\mathrm{KO}}$ with the optimal cutoff). This is not surprising, as $\mathcal{A}_{\mathrm{ideal}}$ decomposes the sampling task into logarithmically many intervals, and the heaviest interval (in terms of the second moment) dominates when logarithmic factors are not considered. Informing the algorithm of the correct cutoff is essentially equivalent to specifying the last heaviest interval. Both of these algorithms have naturally better sample complexity than the other two.

Interestingly, due to $1+\delta < \alpha$, the algorithm $\mathcal{QMC}(\epsilon, 0.8)$ exhibits worse sampling complexity than $\mathcal{A}_{\mathrm{KO}}(\epsilon, \epsilon^{-\frac{1}{\delta}})$, even when 
both algorithms just run \cite{kothari2023mean}
on the truncated version of the random variable with the same truncation threshold. This is because $\mathcal{A}_{\mathrm{KO}}(\epsilon, \epsilon^{-\frac{1}{\delta}})$ is informed of the second moment, whereas $\mathcal{QMC}(\epsilon, 0.8)$ must rely solely on an estimate provided by Theorem \ref{thm:quantum-heavy}. It has to use this estimate because it only has information about the $(1+\delta)$-moment. This illustrates how our knowledge about the variable influences the sample complexity when our goal is to achieve a fixed additive error.

\bibliographystyle{alpha}
\bibliography{ref}

\newcommand{\etalchar}[1]{$^{#1}$}
\begin{thebibliography}{BvDJ{\etalchar{+}}20}

\bibitem[AG07]{Asmussen2007StochasticSA}
S{\o}ren Asmussen and Peter~W. Glynn.
\newblock Stochastic simulation: Algorithms and analysis.
\newblock 2007.

\bibitem[BBBV97]{Bennett1997Strengths}
Charles~H. Bennett, Ethan Bernstein, Gilles Brassard, and Umesh Vazirani.
\newblock Strengths and weaknesses of quantum computing.
\newblock {\em SIAM Journal on Computing}, 26(5):1510--1523, 1997.

\bibitem[Bel19]{Belovs2019QuantumAF}
Aleksandrs Belovs.
\newblock Quantum algorithms for classical probability distributions.
\newblock In {\em Embedded Systems and Applications}, 2019.

\bibitem[BG81]{Bennett1981Relative}
Charles~H. Bennett and John Gill.
\newblock Relative to a random oracle a, \(p^a \neq np^a \neq co-np^a\) with
  probability 1.
\newblock {\em SIAM Journal on Computing}, 10(1):96--113, 1981.

\bibitem[BHMT02]{brassard2002quantum}
Gilles Brassard, Peter Hoyer, Michele Mosca, and Alain Tapp.
\newblock Quantum amplitude amplification and estimation.
\newblock {\em Contemporary Mathematics}, 305:53--74, 2002.

\bibitem[BHT98]{brassard1998quantum}
Gilles Brassard, Peter H{\o}yer, and Alain Tapp.
\newblock Quantum counting.
\newblock In {\em Automata, Languages and Programming: 25th International
  Colloquium, ICALP'98 Aalborg, Denmark, July 13--17, 1998 Proceedings 25},
  pages 820--831. Springer, 1998.

\bibitem[BvDJ{\etalchar{+}}20]{bouland2020prospects}
Adam Bouland, Wim van Dam, Hamed Joorati, Iordanis Kerenidis, and Anupam
  Prakash.
\newblock Prospects and challenges of quantum finance.
\newblock {\em arXiv preprint arXiv:2011.06492}, 2020.

\bibitem[CLL{\etalchar{+}}22]{childs2022quantum}
Andrew~M Childs, Tongyang Li, Jin-Peng Liu, Chunhao Wang, and Ruizhe Zhang.
\newblock Quantum algorithms for sampling log-concave distributions and
  estimating normalizing constants.
\newblock {\em Advances in Neural Information Processing Systems},
  35:23205--23217, 2022.

\bibitem[DLLO16]{devroye2016sub}
Luc Devroye, Matthieu Lerasle, Gabor Lugosi, and Roberto~I Olivekra.
\newblock Sub-gaussian mean estimators.
\newblock {\em The Annals of Statistics}, 44(6):2695--2725, 2016.

\bibitem[GDM{\etalchar{+}}01]{gencay2001}
R.~Gencay, M.~Dacorogna, U.~Muller, O.~Pictet, and R.~Olsen.
\newblock {\em An Introduction to High-Frequency Finance}.
\newblock Academic Press, London, 2001.

\bibitem[Gro96]{grover1996fast}
Lov~K Grover.
\newblock A fast quantum mechanical algorithm for database search.
\newblock In {\em Proceedings of the twenty-eighth annual ACM symposium on
  Theory of computing}, pages 212--219, 1996.

\bibitem[GS10]{GS2010}
Michael Grabchak and Gennady Samorodnitsky.
\newblock Do financial returns have finite or infinite variance? a paradox and
  an explanation.
\newblock {\em Quantitative Finance}, 10(8):883--893, 2010.

\bibitem[Gut05]{gut2005probability}
Allan Gut.
\newblock {\em Probability: a graduate course}, volume 200.
\newblock Springer, 2005.

\bibitem[Ham21]{Hamoudi2021Quantum}
Yassine Hamoudi.
\newblock {\em Quantum Algorithms for the Monte Carlo Method}.
\newblock PhD thesis, Universit{\'e} de Paris, 2021.

\bibitem[Hei03]{heinrich2003monte}
Stefan Heinrich.
\newblock From monte carlo to quantum computation.
\newblock {\em Mathematics and Computers in Simulation}, 62(3-6):219--230,
  2003.

\bibitem[HM19]{hamoudi2019quantum}
Yassine Hamoudi and Fr{\'e}d{\'e}ric Magniez.
\newblock Quantum {Chebyshev's Inequality and Applications}.
\newblock In {\em 46th International Colloquium on Automata, Languages, and
  Programming (ICALP 2019)}, 2019.

\bibitem[iS99]{Sato1999LvyPA}
Ken iti Sato.
\newblock L{\'e}vy processes and infinitely divisible distributions.
\newblock 1999.

\bibitem[JKB94]{johnson1994continuous}
Norman~L. Johnson, Samuel Kotz, and N.~Balakrishnan.
\newblock {\em Continuous Univariate Distributions, Volume 1}.
\newblock Wiley, 1994.

\bibitem[KO23]{kothari2023mean}
Robin Kothari and Ryan O'Donnell.
\newblock Mean estimation when you have the source code; or, quantum {Monte
  Carlo} methods.
\newblock In {\em Proceedings of the 2023 Annual ACM-SIAM Symposium on Discrete
  Algorithms (SODA)}, pages 1186--1215. SIAM, 2023.

\bibitem[Man63]{mandelbrot1963}
B.~Mandelbrot.
\newblock The variation of certain speculative prices.
\newblock {\em Journal of Business}, 26:394--419, 1963.

\bibitem[Mon15]{montanaro2015quantum}
Ashley Montanaro.
\newblock {Quantum speedup of Monte Carlo methods}.
\newblock {\em Proceedings of the Royal Society A: Mathematical, Physical and
  Engineering Sciences}, 471(2181):20150301, 2015.

\bibitem[PT17]{pipiras_taqqu_2017}
Vladas Pipiras and Murad~S. Taqqu.
\newblock {\em Long-Range Dependence and Self-Similarity}.
\newblock Cambridge Series in Statistical and Probabilistic Mathematics.
  Cambridge University Press, 2017.

\bibitem[TTW95]{taqqu1995estimators}
Murad~S. Taqqu, Vadim Teverovsky, and Walter Willinger.
\newblock Estimators for long-range dependence: An empirical study.
\newblock {\em Fractals}, 03(04):785--798, 1995.

\bibitem[Wat19]{Watkins_2019}
N.~W. Watkins.
\newblock Mandelbrot's stochastic time series models.
\newblock {\em Earth and Space Science}, 6(11):2044--2056, 2019.

\end{thebibliography}

\end{document}